\documentclass[oneside, 11pt, a4paper]{article}

\usepackage{amsmath}
\usepackage{amssymb}
\usepackage{amsthm}
\usepackage{natbib}
\usepackage{lmodern}

\newtheorem{theorem}{Theorem}
\newtheorem{lemma}[theorem]{Lemma}

\theoremstyle{remark}
\newtheorem*{remark}{Remark}

\newcommand{\norm}[1]{\lVert #1 \rVert}
\newcommand{\inprod}[2]{\left<{#1},{#2}\right>}

\begin{document}
\title{Objective and subjective entropy measures of portfolio suboptimality}
\author{Ati S. Sharma\thanks{\small ati@agalmic.ltd}}
\date{9 July 2026}

\maketitle

%--------------------------------------------------------------------

The Kelly betting strategy maximises the asymptotic exponential growth rate of a bettor's capital by staking a fraction of it, in proportion to the bettor's assessed probability \citep{Kelly_1956}.
Kelly showed that for repeated, compounding bets, the maximal growth rate is equal to the mutual information between the bettor's side channel and the true outcome \citep{Cover_1991}.
Adopted in finance by \citet{Latane_1959}, the analogous criterion has been studied extensively in continuous time \citep{Merton_1969, Thorp_1975}, where the Kelly portfolio maximises the expected drift of log-wealth.
The cost of deviating from this optimum, which we term the growth gap, is usually presented as an $L^2$ shortfall in the market price of risk.

We show that the same quantity admits two exact relative-entropy representations.
Under the true measure~$P$, the expected log-wealth shortfall equals the KL divergence from~$P$ to the tilted measure~$P^\pi$ under which~$\pi$ is optimal. 
Under~$P^\pi$, the same portfolio appears to outperform the Kelly portfolio, and the apparent outperformance equals the reverse KL divergence.

Consider a filtered probability space $(\Omega, \mathcal{F}, \mathbb{F}, P)$ with $\mathbb{F} = (\mathcal{F}_t)_{0 \leq t \leq T}$ satisfying the usual conditions \citep{Karatzas_1998, Oksendal_2003}, carrying an $m$-dimensional Brownian motion $W = (W_t)_{0 \leq t \leq T}$.

The market consists of a risk-free money-market account $B_t = \exp\left(\int_0^t r_s  ds\right)$ with short rate process $r_t$, and $d$ risky assets with discounted price vector $S_t \in \mathbb{R}^d$ satisfying
\[
    \frac{dS_t}{S_t} = r_t \mathbf{1}  dt + \mu_{e,t}  dt + \sigma_t  dW_t,
\]
where $\mu_{e,t} \in \mathbb{R}^d$ is the excess price-return drift, $\sigma_t \in \mathbb{R}^{d \times m}$ has full row rank a.s., and $\mathbf{1} \in \mathbb{R}^d$ is a vector of ones.

The market price of risk $\theta_t \in \mathbb{R}^m$ is defined by $\mu_{e,t} = \sigma_t \theta_t$.
For an admissible portfolio $\pi_t \in \mathbb{R}^d$, the undiscounted wealth process $X_t^\pi$ satisfies
\[
    \frac{dX_t^\pi}{X_t^\pi} = r_t  dt + \pi_t^\intercal \mu_{e,t}  dt + \pi_t^\intercal \sigma_t  dW_t.
\]
Define the undiscounted log-wealth $Y_t^\pi = \log X_t^\pi$.
Assume standard integrability conditions hold (e.g., Novikov condition) ensuring all It\^o integrals are martingales and expectations are finite.

\begin{theorem}[Growth of the Kelly portfolio]
\label{thm:Kelly growth}
The instantaneous growth rate (drift of $Y_t^\pi$) is
\[
    g_t^\pi = r_t + \inprod{ \sigma_t^\intercal \pi_t }{ \theta_t } - \frac{1}{2} \norm{\sigma_t^\intercal \pi_t}^2.
\]
The growth-optimal (Kelly) portfolio $\pi_t^\star$ for which the growth rate is maximised is characterised by $\sigma_t^\intercal \pi_t^\star = \theta_t$, and its growth rate is
\begin{equation}
    \label{eqn:Kelly growth}
    g_t^\star = r_t + \frac{1}{2} \norm{\theta_t}^2.
\end{equation}
\end{theorem}

\begin{remark}
The equality $g^\star_t = r_t + \frac{1}{2}\|\theta_t\|^2$ presumes the canonical minimum-norm market price of risk, for which $\theta_t \in \operatorname{range}(\sigma_t^\intercal)$.
\end{remark}

\begin{proof}
Applying It\^o's lemma to $Y_t^\pi = \log X_t^\pi$,
\[
    dY_t^\pi = \frac{dX_t^\pi}{X_t^\pi} - \frac{1}{2} \frac{d[X^\pi]_t}{(X_t^\pi)^2}.
\]
The quadratic variation term is $\frac{d[X^\pi]_t}{(X_t^\pi)^2} = \norm{\sigma_t^\intercal \pi_t}^2  dt$.
Substituting $\mu_{e,t} = \sigma_t \theta_t$ yields
\[
    dY_t^\pi
    = \left[
        r_t
        + \inprod{ \sigma_t^\intercal \pi_t }{ \theta_t }
        - \frac{1}{2} \norm{\sigma_t^\intercal \pi_t}^2
    \right] dt
    + \pi_t^\intercal \sigma_t dW_t.
\]
Thus, $g_t^\pi = r_t + \inprod{ \sigma_t^\intercal \pi_t }{ \theta_t } - \frac{1}{2} \norm{\sigma_t^\intercal \pi_t}^2$.
Maximisation over $\pi_t$ is equivalent to maximising $\inprod{ z_t }{ \theta_t } - \frac{1}{2} \|z_t\|^2$ for $z_t = \sigma_t^\intercal \pi_t$, which is solved by $z_t = \theta_t$, giving $g_t^\star = r_t + \frac{1}{2} \norm{\theta_t}^2$.
The quadratic variation term is $\|\sigma_t^\intercal \pi_t\|^2\,dt$ because cross-variation terms vanish for the scalar wealth process (see \citet{Oksendal_2003,Karatzas_1998}).
\end{proof}

\begin{remark}[Leverage constraint and long-only]
The closed-form solution can produce weights whose absolute values sum to more than 1 (leverage) or negative weights (shorting).
Once we impose $\pi \geq 0,\  \mathbf{1}^\intercal \pi \leq 1$, the problem becomes a quadratic program with box constraints.
\end{remark}

Next, we see that the growth gap of a suboptimal portfolio $\pi$ equals the relative entropy between $P$ and a hypothetical measure $P^\pi$ under which $\pi$ would be optimal.

\begin{lemma}[Quadratic growth gap of suboptimal portfolio]
\label{lem:growth gap}
For any admissible portfolio $\pi$, the pointwise growth gap is
\[
    g_t^\star - g_t^\pi = \frac{1}{2} \norm{\theta_t - \sigma_t^\intercal \pi_t}^2 \geq 0.
\]
Consequently, the integrated expectation satisfies
\[
    \mathbb{E}_P\left[ \log \frac{X_T^{\pi^\star}}{X_T^\pi} \right]
    = \frac{1}{2} \mathbb{E}_P\left[ \int_0^T \norm{\theta_s - \sigma_s^\intercal \pi_s}^2  ds \right].
\]
\end{lemma}

\begin{proof}
The pointwise identity follows from completing the square,
\begin{align*}
    g_t^\star - g_t^\pi
    &=   \left( r_t + \frac{1}{2} \norm{\theta_t}^2 \right)
        - \left( r_t + \inprod{ \sigma_t^\intercal \pi_t }{ \theta_t } - \frac{1}{2} \norm{\sigma_t^\intercal \pi_t}^2 \right) \\
    &= \frac{1}{2} \norm{\theta_t - \sigma_t^\intercal \pi_t}^2.
\end{align*}
Integrating and taking expectations gives
\begin{align*}
    \mathbb{E}_P\left[ Y_T^{\pi^\star} - Y_T^\pi \right]
    &= \mathbb{E}_P\left[ \int_0^T (g_t^\star - g_t^\pi)  dt \right] + \mathbb{E}_P\left[ \int_0^T (\pi_t^\star - \pi_t)^\intercal \sigma_t  dW_t \right].
\end{align*}
The stochastic integral has mean zero under the admissibility and integrability conditions, yielding the result.
See \citet{Platen_2006} for the num\'eraire property and supermartingale argument.
\end{proof}

\begin{theorem}[Objective suboptimality as a relative entropy]
    \label{thm:objective growth gap}
    Let $\pi^\star$ be the growth-optimal portfolio and $\pi \neq \pi^\star$ be a suboptimal portfolio.
    Let measure $P^\pi$ be a measure under which the volatility matrix $\sigma_t$ coincides with that under $P$ and $\pi$ is growth-optimal.
    The growth gap (under $P$) between $\pi^\star$ and $\pi$ is equal to the relative entropy of $P$ to $P^\pi$,
    \begin{equation}
        \mathbb{E}_P\left[ \log \frac{X_T^{\pi^\star}}{X_T^\pi} \right] = D_{\mathrm{KL}}(P \parallel P^\pi)
    \end{equation}
\end{theorem}

\begin{proof}
First we introduce a tilted measure under which $\pi$ is optimal.
Under the physical measure $P$, the market price of risk is $\theta_t$.
For portfolio $\pi$, define its associated (implied) market price of risk,
\begin{equation}
    \label{eqn:portfolio}
    \theta_t^\pi := \sigma_t^\intercal \pi_t.
\end{equation}
We define the measure $P^\pi$ via Girsanov's theorem \citep{Karatzas_1998, Oksendal_2003, Leonard_2012}. Lemma~\ref{lem:growth gap} supplies the finite-entropy condition. Under $P^\pi$, the asset dynamics have market price of risk $\theta_t^\pi$,
\[
  \frac{dS_t}{S_t}
  = r_t \mathbf{1}  dt + \sigma_t \theta_t^\pi  dt + \sigma_t  dW_t^\pi,
\]
where  $W^\pi_t = W_t + \int_0^t (\theta_s - \theta_s^\pi) ds$ is a $P^\pi$-Brownian motion,
with Radon-Nikodym derivative
\[
  \frac{dP^\pi}{dP}\bigg|_{\mathcal{F}_T}
  = \exp\left( -\int_0^T (\theta_s - \theta_s^\pi)^\intercal  dW_s - \frac{1}{2} \int_0^T \norm{\theta_s - \theta_s^\pi}^2  ds \right).
\]
Under $P^\pi$, the growth rate of portfolio $\pi$ is
\[
    g_t^{\pi}(P^\pi)
    = r_t + \inprod{ \sigma_t^\intercal \pi_t }{ \theta_t^\pi } - \frac{1}{2} \|\sigma_t^\intercal \pi_t\|^2
    = r_t + \frac{1}{2} \|\theta_t^\pi\|^2,
\]
which is the optimal growth rate under $P^\pi$ because $\sigma_t^\intercal \pi_t = \theta_t^\pi$ by \eqref{eqn:portfolio}.
The relative entropy is
\[
    D_{\mathrm{KL}}(P \parallel P^\pi)
    = \mathbb{E}_P\left[ -\log \frac{dP^\pi}{dP} \right].
\]
Substituting the Radon-Nikodym derivative
and taking expectations (the It\^o integral is a $P$-martingale with mean zero) gives
\begin{equation}
    \label{eqn:path-space-kl}
    D_{\mathrm{KL}}(P \parallel P^\pi)
    = \frac{1}{2} \mathbb{E}_P\left[ \int_0^T \|\theta_s - \theta_s^\pi\|^2  ds \right].
\end{equation}

The growth gap under $P$ (using Lemma~\ref{lem:growth gap}) is
\begin{align*}
    \mathbb{E}_P\left[ \log \frac{X_T^{\pi^\star}}{X_T^\pi} \right]
    &= \frac{1}{2} \mathbb{E}_P\left[ \int_0^T \|\theta_s - \sigma_s^\intercal \pi_s\|^2  ds \right] \\
    &= \frac{1}{2} \mathbb{E}_P\left[ \int_0^T \|\theta_s - \theta_s^\pi\|^2  ds \right] \\
    &= D_{\mathrm{KL}}(P \parallel P^\pi).
\end{align*}
The equality holds provided the Novikov condition $\mathbb{E}_P[\exp(\frac{1}{2}\int_0^T \|\theta_s\|^2\,ds)] < \infty$ is satisfied, ensuring the stochastic integral is a true martingale.

\end{proof}

\begin{theorem}[Subjective outperformance as a relative entropy]
    \label{thm:subjective growth gap}
    Let $\pi^\star$, $\pi$ and $P^\pi$ be defined as in Theorem~\ref{thm:objective growth gap}.
    Under the subjective measure $P^\pi$, the portfolio $\pi$ outperforms $\pi^\star$ by the relative entropy of $P^\pi$ to $P$
    \begin{equation}
        \mathbb{E}_{P^\pi}\left[ \log \frac{X_T^\pi}{X_T^{\pi^\star}} \right] = D_{\mathrm{KL}}(P^\pi \parallel P)
    \end{equation}
\end{theorem}

\begin{proof}
Under $P^\pi$, the portfolio $\pi$ is growth-optimal by construction. We compute the expected log-wealth ratio directly using the drift adjustment under Girsanov.

Recall that $W_t^\pi = W_t + \int_0^t (\theta_s - \theta_s^\pi) ds$ is a $P^\pi$-Brownian motion, so $dW_t = dW_t^\pi - (\theta_t - \theta_t^\pi) dt$. Substituting into the log-wealth dynamics
\begin{align*}
    dY_t^{\pi^\star} &= \left( r_t + \theta_t^\intercal \theta_t^\pi - \tfrac{1}{2}\|\theta_t\|^2 \right) dt + \theta_t^\intercal dW_t^\pi, \\
    dY_t^\pi &= \left( r_t + \tfrac{1}{2}\|\theta_t^\pi\|^2 \right) dt + \theta_t^{\pi\intercal} dW_t^\pi.
\end{align*}
The drift of $Y_t^\pi$ under $P^\pi$ is the optimal growth rate for that measure. Using the drifts from the SDEs above, the gap is
\[
    \left( r_t + \theta_t^\intercal \theta_t^\pi - \tfrac{1}{2}\|\theta_t\|^2 \right) - \left( r_t + \tfrac{1}{2}\|\theta_t^\pi\|^2 \right) = -\tfrac{1}{2}\|\theta_t - \theta_t^\pi\|^2.
\]
Integrating and taking expectations under $P^\pi$ (the stochastic integral with respect to $W^\pi$ has mean zero)
\[
    \mathbb{E}_{P^\pi}\left[ Y_T^{\pi^\star} - Y_T^\pi \right] = -\frac{1}{2} \mathbb{E}_{P^\pi}\left[ \int_0^T \|\theta_s - \theta_s^\pi\|^2 ds \right].
\]
Changing the sign of the path-space drift density under Girsanov gives $D_{\mathrm{KL}}(P^\pi \parallel P) = \frac{1}{2} \mathbb{E}_{P^\pi}\left[ \int_0^T \|\theta_s - \theta_s^\pi\|^2 ds \right]$, whence
\[
    \mathbb{E}_{P^\pi}\left[ Y_T^{\pi^\star} - Y_T^\pi \right] = -D_{\mathrm{KL}}(P^\pi \parallel P).
\]
Therefore
\[
    \mathbb{E}_{P^\pi}\left[ \log \frac{X_T^\pi}{X_T^{\pi^\star}} \right] = D_{\mathrm{KL}}(P^\pi \parallel P). \qedhere
\]
\end{proof}

\begin{remark}[Asymmetry of objective and subjective perspectives]
Theorems~\ref{thm:objective growth gap} and~\ref{thm:subjective growth gap} form a dual pair
\begin{align*}
    \text{Objective:} \quad & \mathbb{E}_P\left[ \log \frac{X_T^{\pi^\star}}{X_T^\pi} \right] = D_{\mathrm{KL}}(P \parallel P^\pi) \geq 0, \\
    \text{Subjective:} \quad & \mathbb{E}_{P^\pi}\left[ \log \frac{X_T^\pi}{X_T^{\pi^\star}} \right] = D_{\mathrm{KL}}(P^\pi \parallel P) \geq 0.
\end{align*}
Both right-hand sides are non-negative. Under the true measure $P$, the optimal portfolio $\pi^\star$ outperforms; under the subjective measure $P^\pi$, the suboptimal portfolio $\pi$ appears to outperform.
\end{remark}

\begin{remark}[Equality of KL divergences for deterministic market price of risk]
When the market price of risk $\theta_s$ and the portfolio-implied price of risk $\theta_s^\pi$ are deterministic functions of time (possibly time-varying but non-random), the relative entropy in both directions coincides
\[
D_{\mathrm{KL}}(P \parallel P^\pi) = D_{\mathrm{KL}}(P^\pi \parallel P) = \frac{1}{2}\int_0^T \|\theta_s - \theta_s^\pi\|^2 ds.
\]
In particular, when both are constant ($\theta_s = \theta$, $\theta_s^\pi = \theta^\pi$ for all $s$), we have
\[
D_{\mathrm{KL}}(P \parallel P^\pi) = D_{\mathrm{KL}}(P^\pi \parallel P) = \frac{T}{2}\|\theta - \theta^\pi\|^2.
\]
This equality holds because deterministic processes have identical distributions under $P$ and $P^\pi$.
When $\theta_s$ or $\theta_s^\pi$ are stochastic, the two KL divergences generally differ.
\end{remark}

\bibliographystyle{abbrvnat}
\bibliography{bibliography}

\end{document}